\documentclass[journal]{IEEEtranTIE}
\usepackage{graphicx}
\pdfoutput=1
\usepackage{cite}
\usepackage{picinpar}
\usepackage{amsmath}
\usepackage{amssymb}
\usepackage{url}
\usepackage{flushend}
\usepackage[latin1]{inputenc}
\usepackage{colortbl}
\usepackage{soul}
\usepackage{multirow}
\usepackage{pifont}
\usepackage{color}
\usepackage{alltt}
\usepackage[hidelinks]{hyperref}
\usepackage{enumerate}
\usepackage{siunitx}
\usepackage{breakurl}
\usepackage{epstopdf}
\usepackage{pbox}
\usepackage{subfigure}
\usepackage{balance}

\newtheorem{assumption}{Assumption}
\newtheorem{lemma}{Lemma}
\newtheorem{remark}{Remark}

\newtheorem{theorem}{Theorem}

\newenvironment{proof}{{\indent \indent \it Proof:}}

\begin{document}
\title{	Globally Composite-Learning-Based Intelligent Fast Finite-Time Control for Uncertain Strict-Feedback Systems with Nonlinearly Periodic Disturbances}

\author{
	\vskip 1em
	
	Xidong Wang, Zhan Li, Zhen He

	\thanks{
					
		Xidong Wang, Zhan Li are with the Research Institute of Intelligent Control and Systems, Harbin Institute of Technology, Harbin 150001, China (e-mail: 17b904039@stu.hit.edu.cn; zhanli@hit.edu.cn). (Corresponding author: Zhan Li)

        Zhen He is with the Department of Control Science and Engineering, Harbin Institute of Technology, Harbin 150001, China (e-mail: hezhen@hit.edu.cn).

        Zhan Li made contributions to formula derivation and original draft review, and Zhen He made contributions to simulation validation.

	}
}

\maketitle
	
\begin{abstract}
This brief aims at the issue of globally composite-learning-based neural fast finite-time (F-FnT) tracking control for a class of uncertain systems in strict-feedback form subject to nonlinearly periodic disturbances. First, uncertain dynamics with periodic parameters are identified by incorporating Fourier series expansion (FSE) into an intelligent estimator, which leverages the feedback of newly designed prediction errors in updating weights to boost learning performance. Then, a novel switching mechanism is constructed to fulfill smooth switching from the composite FSE-based neural controller to robust control law when the inputs of the intelligent estimator transcend the valid approximation domain. By fusing the switching mechanism with an improved F-FnT backstepping algorithm, the globally F-FnT boundedness of all variables in the closed-loop system is guaranteed. Finally, a simulation study is conducted to evince the availability of the theoretical result.
\end{abstract}


{}

\definecolor{limegreen}{rgb}{0.2, 0.8, 0.2}
\definecolor{forestgreen}{rgb}{0.13, 0.55, 0.13}
\definecolor{greenhtml}{rgb}{0.0, 0.5, 0.0}

\section{Introduction}

\IEEEPARstart{O}{ver} the past few decades, adaptive backstepping algorithms have gained a great deal of academic attention owing to their resultful abilities to deal with tracking control issues for strict-feedback nonlinear systems \cite{ACFB}. To obtain a rapid transient response and strong robustness in practical tracking control, the finite-time (FnT) backstepping technique was developed \cite{FnT}. On account of the universal approximation capability, fuzzy logic systems (FLSs) and neural networks (NNs) were exploited into adaptive backstepping algorithms to handle intricate uncertain dynamics \cite{fuzzy2,neural1}. To enhance learning performance, the prediction errors obtained from the serial-parallel estimation model were entered into the neural update laws \cite{composite1,composite4}. It is worth noting that the prerequisite for excellent estimation performance is that intelligent approximators should retain valid at all times. Once a transient transcends the active region of the intelligent approximator, tracking performance slumps, and even the system becomes unstable. Therefore, to achieve global stabilization of intelligent control, switching schemes along with robust controllers were established to indicate the intelligent approximation domain, and bring the external transients back into the working region \cite{global2,global4,global5}. 

In practical applications, periodic time-varying (PTV) perturbations are natural in many control systems \cite{PTV}. Unfortunately, when the PTV parameters are displayed nonlinearly in uncertainties, the aforementioned approaches cannot be employed directly to compensate for such unknown items. The authors of \cite{FSE} designed a FSE-based neural approximator to identify uncertain dynamics with PTV parameters in nonlinear form, where the FSE was adopted to model unknown PTV parameters, and then the approximated values were utilized as the inputs signals of radial basis function neural network (RBFNN) to compensate for the unknown terms. Afterwards, numerous modified FSE-based intelligent control protocols were developed to apply to various kinds of control systems \cite{FSE4,FSE5,FSE7}. In \cite{FSE6}, a FSE-RBFNN-based backstepping controller was constructed to fulfill the FnT tracking objective for strict-feedback nonlinear systems in the presence of periodic perturbations. Nonetheless, there exist two issues that are seldom taken into account: i) Most FSE-based intelligent controllers focus on the system stabilization, while neglecting the accuracy of identified unknown nonlinearity; ii) Most FSE-based intelligent estimators assume to operate within the valid region, however, the inputs of intelligent estimators, especially the approximate values of PTV parameters, cannot be known in advance, which poses difficulties in parameter setting for intelligent estimators.

Based on the foregoing discussion, this brief designs a novel globally composite FSE-based neural F-FnT backstepping control protocol to cope with the tracking control issue for uncertain strict-feedback systems with nonlinearly PTV disturbances. The core contributions of this brief are generalized below: 1) By adopting the feedback of newly designed prediction errors in updating weights, the performance of FSE-based intelligent estimator in identifying uncertain dynamics with nonlinearly periodic parameters is enhanced; 2) A novel switching mechanism along with robust control law is constructed to circumvent system instability when the inputs of FSE-based intelligent estimator transcend the approximation valid domain; 3) By incorporating the switching mechanism into a modified FnT backstepping algorithm, all the closed-loop variables are guaranteed to be globally F-FnT bounded.

Notation: $\left\|  \cdot  \right\|$ expresses the Euclidean norm, $\left\|  \cdot  \right\|_F$ refers to the Frobenius norm, and ${\lambda _{\max }}\left(  \cdot  \right)$ represents the maximum eigenvalues of a matrix. ${\mathop{\rm sig}\nolimits} {\left( x \right)^m } = {\left| x \right|^m }{\mathop{\rm sgn}} (x)$ with ${\mathop{\rm sgn}} ( \cdot )$ being standard signum function.
\section{Problem formulation and lemmas}
\subsection{Problem Formulation}
This brief is concerned with the following uncertain strict-feedback systems subject to nonlinearly periodic disturbances
\begin{equation}
\left\{  \begin{aligned}
&{{\dot \eta }_i} = {F_i}\left( {{{\bar \eta }_i},{p_i}\left( t \right)} \right) + {G_i}\left( {{{\bar \eta }_i}} \right){\eta _{i + 1}},i = 1,2, \ldots ,n - 1\\
&{{\dot \eta }_n} = {F_n}\left( {{{\bar \eta }_n},{p_n}\left( t \right)} \right) + {G_n}\left( {{{\bar \eta }_n}} \right)u\\
&y = {\eta _1}
 \end{aligned} \right.
\end{equation}
where ${\bar \eta _i} = {\left[ {{\eta _1},{\eta _2}, \ldots ,{\eta _i}} \right]^T} \in {\mathbb{R}}^i$, $\left( {i = 1,2, \ldots ,n} \right)$, represent the system state vectors. $u \in {\mathbb{R}}$ and $y \in {\mathbb{R}}$ stand for the control input and output signals, respectively. ${p_i}\left( t \right):\left[ {0, + \infty } \right) \to {\mathbb{R}}^{q_i}$ depict unknown continuously PTV perturbations with known periods $T_i$, and ${F_i}\left( {{{\bar \eta }_i},{p_i}\left( t \right)} \right): {\mathbb{R}}^{i + {q_i}} \to {\mathbb{R}}$ denote smooth uncertainties with nonlinearly periodic perturbations. ${G_i}\left(  \cdot  \right)$ are the nonzero known control coefficients satisfying $0 < \underline G < \left| {{G_i}\left(  \cdot  \right)} \right| < \bar G$ (Without loss of generality, assume ${G_i}\left(  \cdot  \right)>0$).
\begin{remark}
In practical control systems, there are two types of periodic perturbations: state-dependent periodic perturbations and time-dependent periodic perturbations. When periodic perturbations are related to the states, the uncertainties become state-dependent unknown dynamics, which can be directly identified by RBFNNs. When periodic perturbations are time-dependent and exist in the uncertain terms with linearly parameterized form, then ${F_i}\left( {{{\bar \eta }_i},{p_i}\left( t \right)} \right) = p_i^T\left( t \right){F_{ti}}\left( {{{\bar \eta }_i}} \right)$, which can be tackled by the approach in \cite{linear1}. Therefore, we consider a more general and complex form with PTV perturbations nonlinearly displayed in uncertain dynamics. Due to the unmeasurable PTV perturbations being the input signals of intelligent approximator, the intelligent approximators in [3]-[9] cannot be directly applied to identify ${F_i}\left( {{{\bar \eta }_i},{p_i}\left( t \right)} \right)$. 
\end{remark}

The control target of this brief is to construct a globally FSE-based neural control scheme with composite learning law such that all the closed-loop signals are globally F-FnT bounded, and the tracking error approaches a sufficiently small region near zero in finite time. 
\begin{assumption}
There exist known smooth positive functions ${\bar F_i}\left( {{{\bar \eta }_i},t} \right)$ such that $\left| {{F_i}\left( {{{\bar \eta }_i},{p_i}\left( t \right)} \right)} \right| \le {\mu _i}{\bar F_i}\left( {{{\bar \eta }_i},t} \right)$, where $\mu _i$ are unknown nonnegative constants.
\end{assumption}
\begin{assumption}
The given desired trajectory ${y_d}(t)$ and its first derivative ${\dot y_d}(t)$ are smooth and bounded.
\end{assumption}
\subsection{Lemmas}
\begin{lemma} [\cite{FnT}]
Considering the system $\dot \eta  = f(\eta )$, if there exists Lyapunov function $V\left( \eta  \right)$ satisfying
$\dot V(\eta ) \le  - {\vartheta _1}V(\eta ) - {\vartheta _2}{V^m}(\eta ){\rm{ + }}{\vartheta _3}$ with ${\vartheta _1} > 0,{\vartheta _2} > 0,0 < {\vartheta _3} < \infty ,0 < m < 1$, then the origin of the system $\dot \eta  = f(\eta )$ is practical F-FnT stable, and the settling time is bounded: 
\begin{equation}
\begin{aligned}
{T_s} \le \max \left\{ {\frac{1}{{\nu {\vartheta _1}\left( {1 - m} \right)}}\ln \frac{{\nu {\vartheta _1}{V^{1 - m}}\left( {{\eta _0}} \right) + {\vartheta _2}}}{{{\vartheta _2}}}},\right.\\
\phantom{=\;\;}\left.{\frac{1}{{{\vartheta _1}\left( {1 - m} \right)}}\ln \frac{{{\vartheta _1}{V^{1 - m}}\left( {{\eta _0}} \right) + \nu {\vartheta _2}}}{{\nu {\vartheta _2}}}} \right\},
\end{aligned}
\end{equation}
where $\nu  \in (0,1)$. Moreover, the residual set can be estimated as
$
\left\{ {\eta :V\left( \eta  \right) \le \min \left\{ {\frac{{{\vartheta _3}}}{{\left( {1 - \nu } \right){\vartheta _1}}},{{\left( {\frac{{{\vartheta _3}}}{{\left( {1 - \nu } \right){\vartheta _2}}}} \right)}^{\frac{1}{m}}}} \right\}} \right\}.
$
\end{lemma}

\begin{lemma} [\cite{lemma2,fix_own}]
For any $\sigma \in {\mathbb{R}}$ and constants $\kappa >0, {\tau _\sigma } >0, {\varepsilon _\sigma } >0$, the following inequalities hold
\begin{equation}
\begin{aligned}
&0 \le \left| \sigma  \right| - \sigma \tanh \left( {\frac{\sigma }{\kappa }} \right) \le 0.2785\kappa,\\
&0 \le \left| \sigma  \right| - {\sigma ^2}\sqrt {\frac{{{\sigma ^2} + \tau _\sigma ^2 + \varepsilon _\sigma ^2}}{{\left( {{\sigma ^2} + \tau _\sigma ^2} \right)\left( {{\sigma ^2} + \varepsilon _\sigma ^2} \right)}}}  < \frac{{{\tau _\sigma }{\varepsilon _\sigma }}}{{\sqrt {\tau _\sigma ^2 + \varepsilon _\sigma ^2} }}.
\end{aligned}
\end{equation}
\end{lemma}

\begin{lemma} [\cite{lemma3}]
For any ${z_c} \in {\mathbb{R}},c = 1,2, \ldots ,M$, $0<\beta \le 1$, one has
${\left( {\sum\limits_{c = 1}^M {\left| {{z_c}} \right|} } \right)^\beta } \le \sum\limits_{c = 1}^M {{{\left| {{z_c}} \right|}^\beta } \le {M^{1 - \beta }}} {\left( {\sum\limits_{c = 1}^M {\left| {{z_c}} \right|} } \right)^\beta }$. 
\end{lemma}

\begin{lemma} [\cite{lemma4}]
 For $\tilde \chi ,\chi  \in {\mathbb{R}}$, ${m_c} = {m_{c2}}/{m_c}_1 < 1$ with ${m_{c2}} > 0,{m_{c1}} > 0$ being odd integers, it holds that $\tilde \chi {\left( {\chi  - \tilde \chi } \right)^{{m_c}}} \le  - {\beta _1}{\tilde \chi ^{1 + {m_c}}} + {\beta _2}{\chi ^{1 + {m_c}}}$, where ${\beta _1} = \frac{1}{{1 + {m_c}}}\left[ {{2^{{m_c} - 1}} - {2^{\left( {{m_c} - 1} \right)\left( {{m_c} + 1} \right)}}} \right],{\beta _2} = \frac{1}{{1 + {m_c}}}\left[ {\frac{{2{m_c} + 1}}{{{m_c} + 1}} + \frac{{{2^{ - {{\left( {{m_c} - 1} \right)}^2}\left( {{m_c} + 1} \right)}}}}{{{m_c} + 1}} - {2^{{m_c} - 1}}} \right]$.
\end{lemma}

To approximate the unknown PTV parameters as the input signals of RBFNN, the following FSE is introduced \cite{FSEB}:
\begin{equation}
p\left( t \right) = {l^T}\rho \left( t \right) + {\tau _p}\left( t \right),\left\| {{\tau _p}\left( t \right)} \right\| \le {\bar \tau _p},
\end{equation}
where $l = {\left[ {{l _1}, \ldots ,{l _q}} \right]} \in {\mathbb{R}}^{m \times q}$, $\rho \left( t \right) = {\left[ {{\rho_1 \left( t \right)}, \ldots ,{\rho_m \left( t \right)}} \right]}^T$ with $\rho_1 \left( t \right)=1,{\rho _{2r}}\left( t \right) = \sqrt 2 \sin \left( {{{2\pi rt} \mathord{\left/
 {\vphantom {{2\pi rt} T}} \right.
 \kern-\nulldelimiterspace} T}} \right),{\rho _{2r + 1}}\left( t \right) = \sqrt 2 \cos \left( {{{2\pi rt} \mathord{\left/
 {\vphantom {{2\pi rt} T}} \right.
 \kern-\nulldelimiterspace} T}} \right),r = 1, \ldots ,{{\left( {m - 1} \right)} \mathord{\left/
 {\vphantom {{\left( {m - 1} \right)} 2}} \right.
 \kern-\nulldelimiterspace} 2}$.

Then the uncertain dynamic $F\left( {\bar \eta ,p\left( t \right)} \right)$ can be identified by RBFNN on a compact domain:
\begin{equation}
F\left( {\bar \eta ,p\left( t \right)} \right) = {\Omega ^T}H\left( {\bar \eta ,{l^T}\rho \left( t \right)} \right) + {\tau _F}\left( {\bar \eta ,t} \right),
\end{equation}
where $\Omega$ stands for the optimal weight vector, $H\left( {\bar \eta ,{l^T}\rho \left( t \right)} \right) = {\left[ {{h_1}\left( {\bar \eta ,{l^T}\rho \left( t \right)} \right), \ldots ,{h_k}\left( {\bar \eta ,{l^T}\rho \left( t \right)} \right)} \right]^T}$ with ${h_j}\left(  \cdot  \right)$ being Gaussian basis functions $\left( {j = 1,2, \ldots ,k} \right)$, and $\left| {{\tau _F}\left( {\bar \eta ,t} \right)} \right| \le {\bar \tau _F}$.

\begin{lemma} [\cite{FSE}]
Define ${\tilde \Omega } = {\Omega } - {\hat \Omega },{\tilde l} = {l} - {\hat l}$, then
\begin{equation}
\begin{aligned}
&{\Omega ^T}H\left( {\bar \eta ,{l^T}\rho \left( t \right)} \right) - {{\hat \Omega }^T}H\left( {\bar \eta ,{{\hat l}^T}\rho \left( t \right)} \right)\\
& = {{\tilde \Omega }^T}\left( {\hat H - \hat H'{{\hat l}^T}\rho \left( t \right)} \right) + {{\hat \Omega }^T}\hat H'{{\tilde l}^T}\rho \left( t \right) + {d_F},
\end{aligned}
\end{equation}
where $\left| {{d_F}} \right| \le {\left\| l \right\|_F}{\left\| {\rho \left( t \right){{\hat \Omega }^T}\hat H'} \right\|_F} + \left\| \Omega  \right\|\left\| {\hat H'{{\hat l}^T}\rho \left( t \right)} \right\| + {\left| \Omega  \right|_1}$, $\hat H = H\left( {\bar \eta ,{{\hat l}^T}\rho \left( t \right)} \right)$, and $\hat H' = {\left[ {{{\hat h'}_1}, \ldots ,{{\hat h'}_k}} \right]^T}$ with ${\hat h'_j} = {{\left( {\partial {h_j}\left( {\bar \eta ,p} \right)} \right)} \mathord{\left/
 {\vphantom {{\left( {\partial {h_j}\left( {\bar \eta ,p} \right)} \right)} {\partial p}}} \right.
 \kern-\nulldelimiterspace} {\partial p}}{|_{p = {{\hat l}^T}\rho \left( t \right)}}$, $\left( {j = 1,2, \ldots ,k} \right)$.
\end{lemma}

To achieve smooth switching, the following $nth$-order continuous differentiable function is introduced \cite{global2}: 
\begin{equation}
{\varpi _j}\left( {{\eta _j}} \right) = 
\left\{
\begin{aligned}
&1,&\left| {{\rho _j}} \right| \le {c_{j1}}\\
&{\cos ^n}\left( {\frac{\pi }{2}{{\sin }^n}\left( {\frac{\pi }{2}\frac{{\eta _j^2 - c_{j1}^2}}{{c_{j2}^2 - c_{j1}^2}}} \right)} \right),&otherwise\\
&0,&\left| {{\rho _j}} \right| \ge {c_{j2}}
\end{aligned} 
\right.
\end{equation}
where $c_{j1}>0$ and $c_{j2}>0$ stand for the left and right boundaries of RBFNNs' valid domain, $j = 1,2, \ldots ,i$.

Then, a novel switch indicator function for identifying whether the inputs of the RBFNN, including state variables and the approximate values of PTV parameters, are located in the active region is designed as follows 
\begin{equation}
{w_i}\left( {{{\bar \eta }_i},{{\hat p}_i}\left( t \right)} \right) = \prod\limits_{j = 1}^i {{\varpi _j}\left( {{\eta _j}} \right) \cdot } \prod\limits_{j = 1}^{{q_i}} {{\varpi _{{p_{ij}}}}\left( {{{\hat p}_{ij}}\left( t \right)} \right)},
\end{equation}
where ${\hat p_i}\left( t \right) = \hat l_i^T{\rho _i}\left( t \right) = {\left[ {{{\hat p}_{i1}}\left( t \right), \ldots ,{{\hat p}_{i{q_i}}}\left( t \right)} \right]^T}$.
\begin{remark}
The switch indicator function is developed to enable smooth switching between the composite FSE-based RBFNN controller and the robust control law, which guarantees that only the composite adaptive law is valid when the state variables and approximate values of PTV parameters are within the RBFNN's working region to save control energy.
\end{remark}
\section{Main Results}
\subsection{Globally composite FSE-based neural FnT control law}
Define tracking errors
\begin{equation}
\begin{aligned}
&{\xi _1} = {\eta _1} - {y_d},
&{\xi _i} = {\eta _i} - {\eta _{i,c}},i = 2,3, \ldots ,n
\end{aligned}
\end{equation}
where ${\eta _{i,c}}$ are attained from the filtering of virtual control inputs ${\alpha _{i - 1}}$ via the following rapid FnT command filter \cite{differentiator1}. 
\begin{equation}
\begin{aligned}
{{\dot \eta }_{i,c}} = &{\eta _{i,d}},\\
\varepsilon _{ci}^2{{\dot \eta }_{i,d}} =&  - {a_{i1}}\left( {{\eta _{i,c}} - {\alpha _{i - 1}}} \right) - {a_{i2}}{\mathop{\rm sig}\nolimits} {\left( {{\eta _{i,c}} - {\alpha _{i - 1}}} \right)^{{m_{ic}}}} \\
&- {b_{i1}}{\varepsilon _{ci}}{\eta _{i,d}} - {b_{i2}}{\mathop{\rm sig}\nolimits} {\left( {{\varepsilon _{ci}}{\eta _{i,d}}} \right)^{{m_{di}}}},
\end{aligned}
\end{equation}
where ${a_{i1}} > 0,{a_{i2}} > 0,{b_{i1}} > 0,{b_{i2}} > 0,{\varepsilon _{ci}} > 0,{m_{di}} \in \left( {0,1} \right),{m_{ic}} \in \left( {{m_{di}}/\left( {2 - {m_{di}}} \right),1} \right)$ are proper tuning parameters.

To alleviate the impact of $\left( {{\eta _{i,c}} - {\alpha _{i - 1}}} \right)$, the following compensation system is adopted \cite{fuzzy2}
\begin{equation}
\begin{aligned}
&{{\dot \delta }_1} =  - {k_1}{\delta _1} + {G_1}\left( {{\eta _{2,c}} - {\alpha _1}} \right) + {G_1}{\delta _2} - {r_1}\delta _1^{m_c},\\
&{{\dot \delta }_i} =  - {k_i}{\delta _i} + {G_i}\left( {{\eta _{\left( {i + 1} \right),c}} - {\alpha _i}} \right) - {G_{i - 1}}{\delta _{i - 1}} + {G_i}{\delta _{i + 1}} - {r_i}\delta _i^{m_c},\\
&{{\dot \delta }_n} =  - {k_n}{\delta _n} - {G_{n - 1}}{\delta _{n - 1}} - {r_n}\delta _n^{m_c},
\end{aligned}
\end{equation}
where ${r_i} > 0,{k_i} > 0,0.5 < {m_c} = {{{m_c}_2} \mathord{\left/
 {\vphantom {{{m_c}_2} {{m_c}_1}}} \right.
 \kern-\nulldelimiterspace} {{m_c}_1}} < 1$ (${{m_c}_2},{{m_c}_1}$ being odd integers) are proper tuning parameters. 

Then, virtual control inputs ${\alpha _i}$ and the controller $u = {\alpha _n}$ are designed as below
\begin{equation}\small
\begin{aligned}
&{\alpha _1} = \frac{1}{{{G_1}}}\left[ { - {k_1}{\xi _1} + {{\dot y}_d} - {n_1}{\psi _1}\left( {{\sigma _1}} \right) - \left( {1 - {w _1}} \right){{\hat \mu }_1}{{\bar F}_1}\tanh \left( {\frac{{{{\bar F}_1}{\sigma _1}}}{{{\kappa _1}}}} \right)} \right]\\
& - \frac{{{w _1}}}{{{G_1}}}\left[ {\hat \Omega _1^T{{\hat H}_1} + \frac{{{\sigma _1}}}{2}\left\| {{\rho _1}\left( t \right)\hat \Omega _1^T{{\hat H'}_1}} \right\|_F^2 + \frac{{{\sigma _1}}}{2}{{\left\| {{{\hat H'}_1}{{\hat l}_1}^T{\rho _1}\left( t \right)} \right\|}^2}} \right],\\
&{\alpha _i} = \frac{1}{{{G_i}}}\left[ { - {k_i}{\xi _i} + {{\dot \eta }_{i,c}} - {n_i}{\psi _i}\left( {{\sigma _i}} \right) - \left( {1 - {w_i}} \right){{\hat \mu }_i}{{\bar F}_i}\tanh \left( {\frac{{{{\bar F}_i}{\sigma _i}}}{{{\kappa _i}}}} \right)} \right]\\
& - \frac{{{w_i}}}{{{G_i}}}\left[ {\hat \Omega _i^T{{\hat H}_i} + \frac{{{\sigma _i}}}{2}\left\| {{\rho _i}\left( t \right)\hat \Omega _i^T{{\hat H'}_i}} \right\|_F^2 + \frac{{{\sigma _i}}}{2}{{\left\| {{{\hat H'}_i}{{\hat l}_i}^T{\rho _i}\left( t \right)} \right\|}^2}} \right] - \frac{{{G_{i - 1}}}}{{{G_i}}}{\xi _{i - 1}},
\end{aligned}
\end{equation} 
where ${n_i} > 0,{\kappa _i} >0$ are design parameters, and ${w_i} \left(  \cdot  \right)$ are switch indicator functions in (8), ${\sigma _i} = {\xi _i} - {\delta _i}$,
\begin{equation}
{\psi _i}\left( {{\sigma _i}} \right) = \sigma _i^{1 + 2{m_c}}\sqrt {\frac{{\sigma _i^{2 + 2{m_c}} + \tau _{\sigma i}^2 + \varepsilon _{\sigma i}^2}}{{\left( {\sigma _i^{2 + 2{m_c}} + \tau _{\sigma i}^2} \right)\left( {\sigma _i^{2 + 2{m_c}} + \varepsilon _{\sigma i}^2} \right)}}}, 
\end{equation}
and ${\hat \Omega}_i,{\hat l}_i,{\hat \mu}_i$ satisfy the following dynamic equations:
\begin{equation}
\begin{aligned}
&{{\dot {\hat \Omega} }_i} = {\Gamma _{{\Omega _i}}}\left[ {{w_i}\left( {{\sigma _i} + {\gamma _{si}}{s_{ni}}} \right)\left( {{{\hat H}_i} - {{\hat H'}_i}\hat l_i^T{\rho _i}\left( t \right)} \right) - {\gamma _{{ _i}}}{{\hat \Omega }_i}} \right],\\
&{{\dot {\hat l}}_i} = {\Gamma _{{l_i}}}\left[ {{w_i}\left( {{\sigma _i} + {\gamma _{si}}{s_{ni}}} \right){\rho _i}\left( t \right)\hat \Omega _i^T{{\hat H'}_i} - {\gamma _{{i}}}{{\hat l}_i}} \right],\\
&{{\dot {\hat \mu }}_i} = {\gamma _{i1}}\left( {1 - {w_i}} \right){{\bar F}_i}{\sigma _i}\tanh \left( {\frac{{{{\bar F}_i}{\sigma _i}}}{{{\kappa _i}}}} \right) - {\gamma _{i2}}{{\hat \mu }_i} - {\gamma _{i3}}\hat \mu _i^{{m_c}},
\end{aligned}
\end{equation}
where ${\Gamma _{{\Omega _i}}} > 0,{\Gamma _{{l_i}}} > 0$ are proper designed gain matrices, and ${\gamma _{si}},{\gamma _{{i}}},{\gamma _{i1}},{\gamma _{i2}},{\gamma _{i3}}$ are positive design parameters.

The following modified FnT serial-parallel approximation models are designed to obtain the prediction errors ${s_{ni}} = {\eta _i} - {\hat \eta _i}$: (define $\eta_{n+1}=u$)
\begin{equation}
\begin{aligned}
 &{{\dot {\hat \eta} }_i}= {w _i}\left[ {\hat \Omega _i^T{{\hat H}_i} + \frac{{{s_{ni}}}}{2}\left\| {{\rho _i}\left( t \right)\hat \Omega _i^T{{\hat H'}_i}} \right\|_F^2 + \frac{{{s_{ni}}}}{2}{{\left\| {{{\hat H'}_i}{{\hat l}_i}^T{\rho _i}\left( t \right)} \right\|}^2}} \right]\\
 &+ {G_i}{\eta _{i + 1}} + {\upsilon _{i1}}{s_{ni}} + {\upsilon _{i2}}s_{ni}^{m_c} + \left( {1 - {w _i}} \right){{\hat \mu }_{ni}}{{\bar F}_i}\tanh \left( {\frac{{{{\bar F}_i}{s_{ni}}}}{{{\kappa _{ni}}}}} \right),
\end{aligned}
\end{equation}
where
\begin{equation}
\begin{aligned}
{\dot {\hat \mu }_{ni}} = {\gamma _{ni1}}\left( {1 - {w_i}} \right){s_{ni}}{\bar F_i}\tanh \left( {\frac{{{{\bar F}_i}{s_{ni}}}}{{{\kappa _{ni}}}}} \right) - {\gamma _{ni2}}{\hat \mu _{ni}} - {\gamma _{ni3}}\hat \mu _{ni}^{m_c},
\end{aligned}
\end{equation}
with $\upsilon _{i1},\upsilon_{i2},{\kappa _{ni}},{\gamma _{ni1}},{\gamma _{ni2}},{\gamma _{ni3}}$ being positive tuning parameters.
\begin{remark}
When calculating the time derivative of the virtual control variables, the positive fractional power of less than one in the virtual control inputs may generate singularity, which can be circumvented by adopting (13).
\end{remark}
\begin{remark}
The prediction errors extracted from the modified FnT serial-parallel approximation models (15) are added to the FSE-based RBFNN updating laws (14) to enhance the learning performance. 
\end{remark}
\begin{remark}
If we select ${\bar F_i}\left( {{{\bar \eta }_i},t} \right) = 1$, the developed control scheme can be applied to handle the situation where the bound of the uncertain functions is unknown. 
\end{remark}
\subsection{Stability analysis}
\begin{theorem}
Considering the uncertain strict-feedback system (1) with nonlinearly periodic disturbances, the rapid FnT command filter (10), the dynamic equations for compensation (11), the  modified FnT serial-parallel approximation models (15), the composite learning laws (14), the virtual control inputs (12), all the closed-loop signals are globally F-FnT bounded, and the tracking error approaches a sufficiently small region around zero in finite time. 
\end{theorem}
\begin{proof}
The Lyapunov function is constructed below
\begin{equation}
\begin{aligned}
V =& \frac{1}{2}\sum\limits_{i = 1}^n {\left( {\sigma _i^2 + \delta _i^2 + \gamma _{i1}^{ - 1}\tilde \mu _i^2 + {\gamma _{si}}\gamma _{ni1}^{ - 1}\tilde \mu _{ni}^2} \right)} \\
 &+ \frac{1}{2}\sum\limits_{i = 1}^n {\left( {\tilde \Omega _i^T\Gamma _{{\Omega _i}}^{ - 1}{{\tilde \Omega }_i} + tr\left\{ {\tilde l_i^T\Gamma _{{l_i}}^{ - 1}{{\tilde l}_i}} \right\} + {\gamma _{si}}s_{ni}^2} \right)}, 
\end{aligned}
\end{equation}
where ${\tilde \mu _i} = {\mu _i} - {\hat \mu _i},{\tilde \mu _{ni}} = {\mu _i} - {\hat \mu _{ni}}$.

By employing \emph{Lemma 2} and \emph{Lemma 5}, one gets
\begin{equation}\small
\begin{aligned}
&\left( {1 - {w_i}} \right){F_i}{\sigma _i} - \left( {1 - {w_i}} \right){\mu _i}{{\bar F}_i}{\sigma _i}\tanh \left( {\frac{{{{\bar F}_i}{\sigma _i}}}{{{\kappa _i}}}} \right) \le 0.2785{\kappa _i}{\mu _i},\\
&{\sigma _i}{d_{{F_i}}} \le \frac{{\sigma _i^2}}{2}\left\| {{\rho _i}\left( t \right)\hat \Omega _i^T{{\hat H'}_i}} \right\|_F^2 + \frac{{\left\| {{l_i}} \right\|_F^2}}{2} + \frac{{\sigma _i^2}}{2}{\left\| {{{\hat H'}_i}{{\hat l}_i}^T{\rho _i}\left( t \right)} \right\|^2}\\
& + \frac{{{{\left\| {{\Omega _i}} \right\|}^2}}}{2} + \sigma _i^2 + \frac{{\left| {{\Omega _i}} \right|_1^2}}{2} + \frac{{\bar \tau _{{F_i}}^2}}{2}.
\end{aligned}
\end{equation}

By adopting (18) and $tr\left\{ {{{\hat l}_i}^T{\rho _i}\left( t \right)\hat \Omega _i^T{{\hat H'}_i}} \right\} = \hat \Omega _i^T{\hat H'_i}{\hat l_i}^T{\rho _i}\left( t \right)$, the derivative of $V$ with respect to time is
\begin{equation}\small
\begin{aligned}
&\dot V \le  - \sum\limits_{i = 1}^n {\left[ {\left( {{k_i} + 1} \right)\sigma _i^2 + {n_i}\sigma _i^{{m_c} + 1} + {k_i}\delta _i^2 + {r_i}\delta _i^{{m_c} + 1} + {r_i}\delta _i^{{m_c}}{\sigma _i}} \right]} \\
& + \sum\limits_{i = 1}^n {\left[ {\frac{{{\gamma _{i2}}}}{{{\gamma _{i1}}}}{{\tilde \mu }_i}{{\hat \mu }_i} + \frac{{{\gamma _{i3}}}}{{{\gamma _{i1}}}}{{\tilde \mu }_i}\hat \mu _i^{{m_c}} + \frac{{{\gamma _{ni2}}{\gamma _{si}}}}{{{\gamma _{ni1}}}}{{\tilde \mu }_{ni}}{{\hat \mu }_{ni}} + \frac{{{\gamma _{ni3}}{\gamma _{si}}}}{{{\gamma _{ni1}}}}{{\tilde \mu }_{ni}}\hat \mu _{ni}^{{m_c}}} \right]} \\
& + \sum\limits_{i = 1}^n {\left[ {{\gamma _i}\tilde \Omega _i^T{{\hat \Omega }_i} + {\gamma _i}tr\left\{ {\tilde l_i^T{{\hat l}_i}} \right\} - {\gamma _{si}}\left( {{\upsilon _{i1}} - 1} \right)s_{ni}^2 - {\gamma _{si}}{\upsilon _{i2}}s_{ni}^{{m_c} + 1}} \right]} \\
& + \sum\limits_{i = 1}^n {\left( {{c_\sigma }{\kappa _i}{\mu _i} + {c_\sigma }{\gamma _{si}}{\kappa _{ni}}{\mu _{ni}} + \frac{{{n_i}{\tau _{\sigma i}}{\varepsilon _{\sigma i}}}}{{\sqrt {\tau _{\sigma i}^2 + \varepsilon _{\sigma i}^2} }}} \right)}  + \sum\limits_{i = 1}^{n - 1} {{G_i}{\delta _i}\left( {{\eta _{\left( {i + 1} \right),c}} - {\alpha _i}} \right)} \\
& + \sum\limits_{i = 1}^n {\frac{{1 + {\gamma _{si}}}}{2}} \left( {\left\| {{l_i}} \right\|_F^2 + {{\left\| {{\Omega _i}} \right\|}^2} + \left| {{\Omega _i}} \right|_1^2 + \bar \tau _{{F_i}}^2} \right),
\end{aligned}
\end{equation}
where ${c_\sigma }=0.2785$.

According to \cite{differentiator1}, there exist ${\lambda _i} > 0\left( {{\lambda _i}{m_{di}} > 2} \right)$ such that $\left| {{\eta _{i,c}} - {\alpha _{i - 1}}} \right| = {\rm O}\left( {\varepsilon _{ci}^{{\lambda _i}{m_{di}}}} \right)$ in finite time $t_1$, where ${\rm O}\left( {\varepsilon _{ci}^{{\lambda _i}{m_{di}}}} \right)$ represents that the estimation errors between ${\eta _{i,c}}$ and ${\alpha _{i - 1}}$ are ${\varepsilon _{ci}^{{\lambda _i}{m_{di}}}}$ order.
 
By applying Young's inequality and \emph{Lemma 4}, we obtain
\begin{equation}\small
\begin{aligned}
&{G_i}{\delta _i}\left( {{\eta _{\left( {i + 1} \right),c}} - {\alpha _i}} \right) \le \frac{1}{2}\delta _i^2 + \frac{1}{2}{{\bar G}^2}{\rm O}\left( {\varepsilon _{ci}^{2{\lambda _i}{m_{di}}}} \right),\\
&{{\tilde \mu }_i}\hat \mu _i^{{m_c}} \le  - {\beta _1}\tilde \mu _i^{{m_c} + 1} + {\beta _2}\mu _i^{{m_c} + 1},\\
&{r_i}\delta _i^{{m_c}}{\sigma _i} \le \frac{{{r_i}}}{{1 + {m_c}}}\sigma _i^{{m_c} + 1} + \frac{{{r_i}{m_c}}}{{1 + {m_c}}}\delta _i^{{m_c} + 1},\\
&{\gamma _i}\tilde \Omega _i^T{{\hat \Omega }_i} \le  - \frac{{{\gamma _i}}}{2}{\left\| {{{\tilde \Omega }_i}} \right\|^2} + \frac{{{\gamma _i}}}{2}{\left\| {{\Omega _i}} \right\|^2},\\
&{\gamma _i}tr\left\{ {\tilde l_i^T{{\hat l}_i}} \right\} \le  - \frac{{{\gamma _i}}}{2}\left\| {{{\tilde l}_i}} \right\|_F^2 + \frac{{{\gamma _i}}}{2}\left\| {{l_i}} \right\|_F^2,\\
& - {\left\| {{{\tilde \Omega }_i}} \right\|^2} \le  - {\left( {{{\left\| {{{\tilde \Omega }_i}} \right\|}^2}} \right)^{\frac{{{m_c} + 1}}{2}}} + \frac{{1 - {m_c}}}{2}{\left( {\frac{{{m_c} + 1}}{2}} \right)^{\frac{{1 + {m_c}}}{{1 - {m_c}}}}},\\
 &- \left\| {{{\tilde l}_i}} \right\|_F^2 \le  - {\left( {\left\| {{{\tilde l}_i}} \right\|_F^2} \right)^{\frac{{{m_c} + 1}}{2}}} + \frac{{1 - {m_c}}}{2}{\left( {\frac{{{m_c} + 1}}{2}} \right)^{\frac{{1 + {m_c}}}{{1 - {m_c}}}}}.
\end{aligned}
\end{equation}

Substituting (20) into (19) and utilizing \emph{Lemma 3}, when $t>t_1$, we gets
\begin{equation}
\dot V(\eta ) \le  - {\vartheta _1}V(\eta ) - {\vartheta _2}{V^{\frac{{1 + {m_c}}}{2}}}(\eta ){\rm{ + }}{\vartheta _3},
\end{equation}
where 
\begin{equation}\small
\begin{aligned}
{\vartheta _1} = \min &\left\{ {\left( {2{k_j} - 1} \right),2{k_n},{\gamma _{i2}},{\gamma _{ni2}},2\left( {{\upsilon _{i1}} - 1} \right),}\right.\\
&\phantom{=\;\;}\left.{
0.5{\gamma _i}\lambda _{\max }^{ - 1}\left( {\Gamma _{{\Omega _i}}^{ - 1}} \right),0.5{\gamma _i}\lambda _{\max }^{ - 1}\left( {\Gamma _{{l_i}}^{ - 1}} \right)} \right\},\\
{\vartheta _2} = \min &\left\{ {\left( {{n_i} - \frac{{{r_i}}}{{1 + {m_c}}}} \right),\frac{{{r_i}}}{{1 + {m_c}}},\frac{{{\gamma _{i3}}{\beta _1}}}{{\gamma _{i1}^{{\textstyle{{1 - {m_c}} \over 2}}}}},\frac{{{\gamma _{ni3}}\gamma _{si}^{{\textstyle{{1 - {m_c}} \over 2}}}{\beta _1}}}{{\gamma _{ni1}^{{\textstyle{{1 - {m_c}} \over 2}}}}},}\right.\\
&\phantom{=\;\;}\left.{
\gamma _{si}^{{\textstyle{{1 - {m_c}} \over 2}}}{\upsilon _{i2}},\frac{{{\gamma _i}}}{4}\lambda _{\max }^{ - {\textstyle{{1 + {m_c}} \over 2}}}\left( {\Gamma _{{\Omega _i}}^{ - 1}} \right),\frac{{{\gamma _i}}}{4}\lambda _{\max }^{ - {\textstyle{{1 + {m_c}} \over 2}}}\left( {\Gamma _{{l_i}}^{ - 1}} \right)} \right\}\cdot{\vartheta _c},
\end{aligned}
\end{equation}
and ${\vartheta _3} = {c_\sigma }\sum\limits_{i = 1}^n {\left( {{\kappa _i}{\mu _i} + {\gamma _{si}}{\kappa _{ni}}{\mu _{ni}}} \right)}  + \sum\limits_{i = 1}^n {\frac{{{n_i}{\tau _{\sigma i}}{\varepsilon _{\sigma i}}}}{{\sqrt {\tau _{\sigma i}^2 + \varepsilon _{\sigma i}^2} }}}  + \sum\limits_{i = 1}^n {\frac{{1 + {\gamma _{si}}}}{2}} \left( {\left\| {{l_i}} \right\|_F^2 + {{\left\| {{\Omega _i}} \right\|}^2} + \left| {{\Omega _i}} \right|_1^2 + \bar \tau _{{F_i}}^2} \right) + \sum\limits_{i = 1}^n {\left( {\frac{{{\gamma _{i2}}}}{{2{\gamma _{i1}}}} + \frac{{{\gamma _{ni2}}{\gamma _{si}}}}{{2{\gamma _{ni1}}}}} \right)\mu _i^2}  + {\beta _2}\sum\limits_{i = 1}^n {\left( {\frac{{{\gamma _{i3}}}}{{{\gamma _{i1}}}} + \frac{{{\gamma _{ni3}}{\gamma _{si}}}}{{{\gamma _{ni1}}}}} \right)\mu _i^{{m_c} + 1}}  + \frac{{{\gamma _i}}}{2}\sum\limits_{i = 1}^n {\left( {\left\| {{l_i}} \right\|_F^2 + {{\left\| {{\Omega _i}} \right\|}^2}} \right)}  + \frac{1}{2}{\bar G^2}\sum\limits_{i = 1}^{n - 1} {{\rm O}\left( {\varepsilon _{ci}^{2{\lambda _i}{m_{di}}}} \right)}  + \frac{{1 - {m_c}}}{4}{\left( {\frac{{{m_c} + 1}}{2}} \right)^{\frac{{1 + {m_c}}}{{1 - {m_c}}}}}\sum\limits_{i = 1}^n {{\gamma _i}}$. $j = 1, \ldots ,n - 1,i = 1, \ldots ,n$, ${\vartheta _c}={2^{{\textstyle{{1 + {m_c}} \over 2}}}}$.

In terms of \emph{Lemma 1}, ${\sigma _i},{\delta _i}$ tend to the following residual sets
\begin{equation}
\begin{aligned}
&\left| {{\sigma _i}} \right| \le \min \left\{ {\sqrt {\frac{{2{\vartheta _3}}}{{\left( {1 - \nu } \right){\vartheta _1}}}} ,\sqrt {2{{\left( {\frac{{{\vartheta _3}}}{{\left( {1 - \nu } \right){\vartheta _2}}}} \right)}^{\frac{2}{{1 + {m_c}}}}}} } \right\},\\
&\left| {{\delta _i}} \right| \le \min \left\{ {\sqrt {\frac{{2{\vartheta _3}}}{{\left( {1 - \nu } \right){\vartheta _1}}}} ,\sqrt {2{{\left( {\frac{{{\vartheta _3}}}{{\left( {1 - \nu } \right){\vartheta _2}}}} \right)}^{\frac{2}{{1 + {m_c}}}}}} } \right\}
\end{aligned}
\end{equation}
in finite time $T_s$, which is expressed as follows
\begin{equation}
\begin{aligned}
{T_s} = {t_1} + \max &\left\{ {\frac{2}{{\nu {\vartheta _1}\left( {1 - {m_c}} \right)}}\ln \frac{{\nu {\vartheta _1}{V^{{\textstyle{{1 - {m_c}} \over 2}}}}\left( 0 \right) + {\vartheta _2}}}{{{\vartheta _2}}},}\right.\\
&\phantom{=\;\;}\left.{\frac{2}{{{\vartheta _1}\left( {1 - {m_c}} \right)}}\ln \frac{{{\vartheta _1}{V^{{\textstyle{{1 - {m_c}} \over 2}}}}\left( 0 \right) + \nu {\vartheta _2}}}{{\nu {\vartheta _2}}}} \right\}.
\end{aligned}
\end{equation}

When $t \ge {T_s}$, the tracking errors can reach
\begin{equation}
\begin{aligned}
\left| {{\xi _i}} \right| \le \left| {{\sigma _i}} \right| + \left| {{\delta _i}} \right| \le \min \left\{ {2\sqrt {\frac{{2{\vartheta _3}}}{{\left( {1 - \nu } \right){\vartheta _1}}}} ,2\sqrt {2{{\left( {\frac{{{\vartheta _3}}}{{\left( {1 - \nu } \right){\vartheta _2}}}} \right)}^{\frac{2}{{1 + {m_c}}}}}} } \right\}.
\end{aligned}
\end{equation}

The proof is completed.
\end{proof}
\begin{figure}[htbp]
    \centering
	\includegraphics[width=0.4\textwidth]{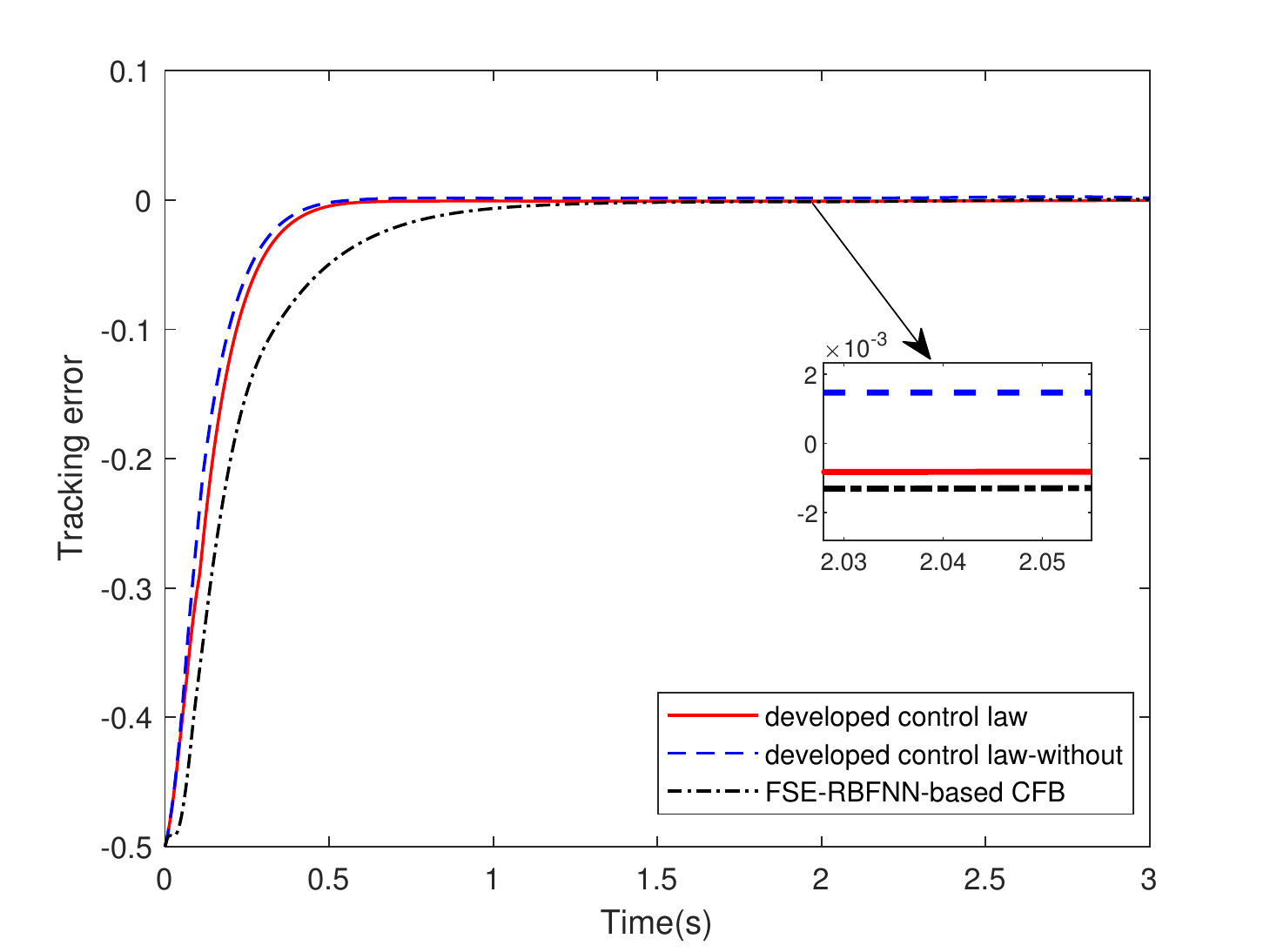}
	\caption{Tracking errors}
\end{figure}

\section{Simulation Results}
The following simulation case \cite{FSE6} is adopted in this section to manifest the availability of the developed control scheme:
\begin{equation}
\begin{aligned}
&{{\dot \eta }_1} = {\eta _2},\\
&{{\dot \eta }_2} = 2.5{\eta _2}\left| {\cos \left( t \right)} \right| - \frac{{0.5MgL}}{J}\sin \left( {{\eta _1}} \right) + \frac{u}{J},
\end{aligned}
\end{equation}
where ${\eta _1}\left( 0 \right) = 0.5, M=2, g=9.8, L=1, J =0.5$, $y_d= \sin (t)$.
\begin{figure}[htbp]
	\centering
	\subfigure[Approximation errors]{
	\includegraphics[width=0.22\textwidth]{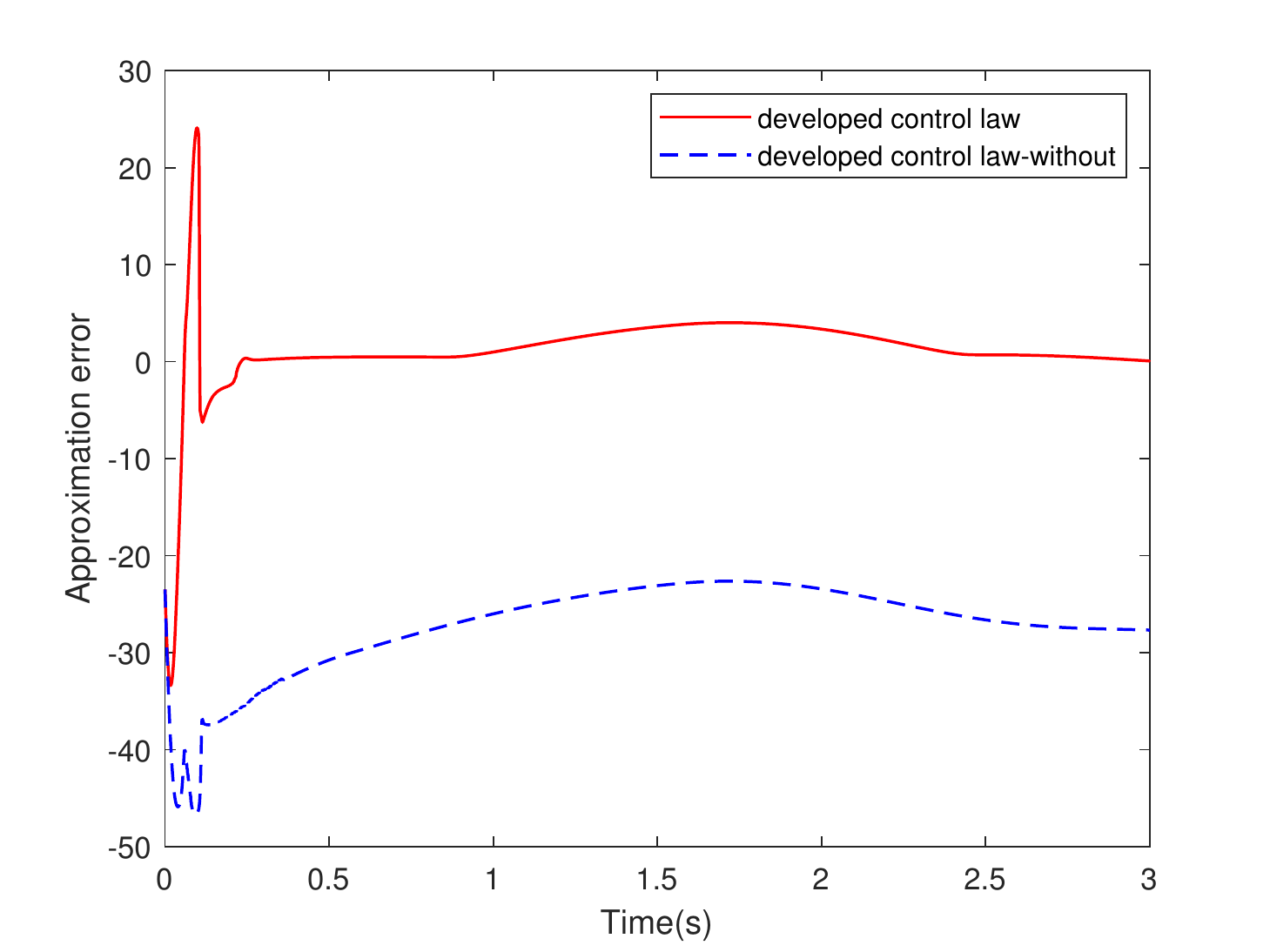}}
	\subfigure[Switching signal]{
	\includegraphics[width=0.22\textwidth]{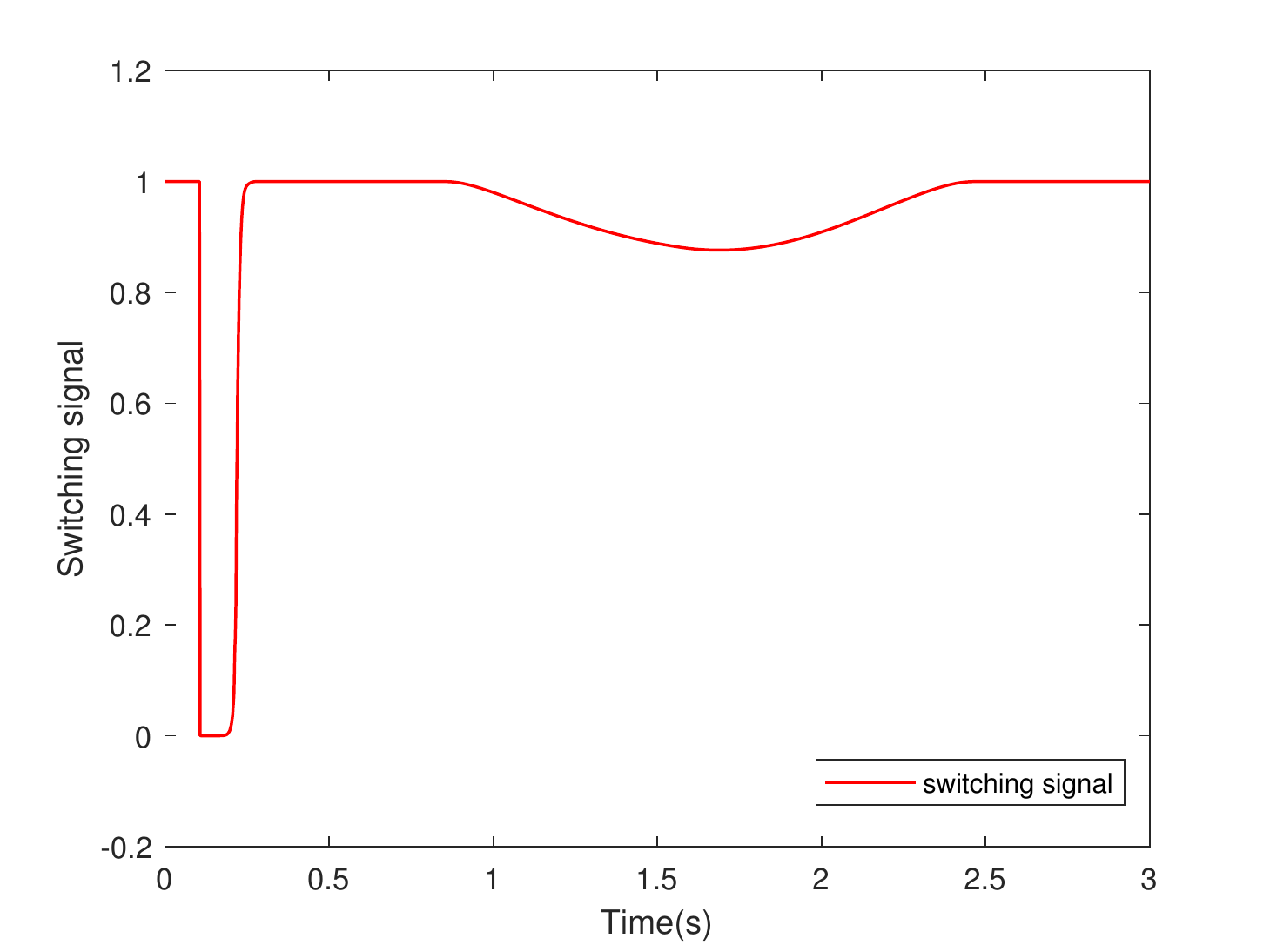}}
	\caption{Approximation errors and Switching signal}
\end{figure}

The parameters setting of the developed control law is ${k_1} = 8,{k_2} = 5,{r_1} = {r_2} = 1,{n_1} = {n_2} = 0.5,{m_c} = 0.6,m = 7,{\Gamma _{{\Omega _2}}} = 10I,{\Gamma _{{l_2}}} = 15I,{\gamma _{s2}} = 5,{\gamma _{21}} = {\gamma _{n21}} = 15,{\gamma _{22}} = {\gamma _{n22}} = 0.001$. The RBFNN embodies 216 nodes with centers distributed evenly in $\left[ { - 1.5,1.5} \right] \times \left[ { - 1.5,1.5} \right] \times \left[ { - 3,3} \right]$ and widths are chosen as 2. The developed control law without composite updating weights and the FSE-RBFNN-based approximator \cite{FSE} combined with the command-filtered backstepping protocol \cite{ACFB} are considered as comparisons, marked as developed control law-without and FSE-RBFNN-based CFB, respectively.

The simulation results depicted in Fig. 1 and Fig. 2 (a) illuminate that the precision of intelligent approximation is heightened by utilizing the designed composite weight updates, which results in superior tracking performance. Additionally, the switching signal in Fig. 2 (b) demonstrates that the developed robust controller can bring the input signals of RBFNN back to the approximation valid region, which guarantees the global stabilization of the closed-loop system.
 
\section{Conclusion}
In this brief, a globally FSE-based neural F-FnT control scheme with composite learning law has been established to address the tracking control issue for a family of nonlinear strict-feedback systems with PTV disturbances. The FSE-based RBFNN with new compound updating law is developed to reinforce the precision of approximation. By combining a novel switching mechanism with a modified F-FnT backstepping protocol, all signals in the closed-loop system are globally F-FnT bounded, and the tracking error approaches a sufficiently small neighborhood of zero in finite time. The availability of the obtained results is exemplified by simulation research.

\balance
\bibliographystyle{Bibliography/IEEEtranTIE}
\bibliography{Bibliography/IEEEabrv,Bibliography/myRef}\ 

\begin{thebibliography}{10}
\providecommand{\url}[1]{#1}
\csname url@samestyle\endcsname
\providecommand{\newblock}{\relax}
\providecommand{\bibinfo}[2]{#2}
\providecommand{\BIBentrySTDinterwordspacing}{\spaceskip=0pt\relax}
\providecommand{\BIBentryALTinterwordstretchfactor}{4}
\providecommand{\BIBentryALTinterwordspacing}{\spaceskip=\fontdimen2\font plus
\BIBentryALTinterwordstretchfactor\fontdimen3\font minus
  \fontdimen4\font\relax}
\providecommand{\BIBforeignlanguage}[2]{{%
\expandafter\ifx\csname l@#1\endcsname\relax
\typeout{** WARNING: IEEEtran.bst: No hyphenation pattern has been}%
\typeout{** loaded for the language `#1'. Using the pattern for}%
\typeout{** the default language instead.}%
\else
\language=\csname l@#1\endcsname
\fi
#2}}
\providecommand{\BIBdecl}{\relax}
\BIBdecl

\bibitem{ACFB}
W.~Dong, J.~A. Farrell, M.~M. Polycarpou, V.~Djapic, and M.~Sharma, ``Command
  filtered adaptive backstepping,'' \emph{IEEE Transactions on Control Systems
  Technology}, vol.~20, no.~3, pp. 566--580, 2012.

\bibitem{FnT}
J.~Yu, P.~Shi, and L.~Zhao, ``Finite-time command filtered backstepping control
  for a class of nonlinear systems,'' \emph{Automatica}, vol.~92, pp. 173--180,
  2018.

\bibitem{fuzzy2}
L.~Zhao, G.~Liu, and J.~Yu, ``Finite-time adaptive fuzzy tracking control for a
  class of nonlinear systems with full-state constraints,'' \emph{IEEE
  Transactions on Fuzzy Systems}, vol.~29, no.~8, pp. 2246--2255, 2021.

\bibitem{neural1}
J.~Qiu, K.~Sun, I.~J. Rudas, and H.~Gao, ``Command filter-based adaptive {NN}
  control for {MIMO} nonlinear systems with full-state constraints and actuator
  hysteresis,'' \emph{IEEE Transactions on Cybernetics}, vol.~50, no.~7, pp.
  2905--2915, 2020.

\bibitem{composite1}
B.~Xu, Z.~Shi, C.~Yang, and F.~Sun, ``Composite neural dynamic surface control
  of a class of uncertain nonlinear systems in strict-feedback form,''
  \emph{IEEE Transactions on Cybernetics}, vol.~44, no.~12, pp. 2626--2634,
  2014.

\bibitem{composite4}
X.~Wang, B.~Xu, Y.~Cheng, H.~Wang, and F.~Sun, ``Robust adaptive learning
  control of space robot for target capturing using neural network,''
  \emph{IEEE Transactions on Neural Networks and Learning Systems}, pp. 1--11,
  2022.

\bibitem{global2}
J.~Wu, W.~Chen, D.~Zhao, and J.~Li, ``Globally stable direct adaptive
  backstepping {NN} control for uncertain nonlinear strict-feedback systems,''
  \emph{Neurocomputing}, vol. 122, pp. 134--147, 2013.

\bibitem{global4}
B.~Xu, X.~Wang, Y.~Shou, P.~Shi, and Z.~Shi, ``Finite-time robust intelligent
  control of strict-feedback nonlinear systems with flight dynamics
  application,'' \emph{IEEE Transactions on Neural Networks and Learning
  Systems}, vol.~33, no.~11, pp. 6173--6182, 2022.

\bibitem{global5}
Q.~Wang, Z.~Zhang, and X.-J. Xie, ``Globally adaptive neural network tracking
  for uncertain output-feedback systems,'' \emph{IEEE Transactions on Neural
  Networks and Learning Systems}, vol.~34, no.~2, pp. 814--823, 2023.

\bibitem{PTV}
M.~Tomizuka, ``Dealing with periodic disturbances in controls of mechanical
  systems,'' \emph{Annual Reviews in Control}, vol.~32, no.~2, pp. 193--199,
  2008.

\bibitem{FSE}
W.~Chen, ``Adaptive backstepping dynamic surface control for systems with
  periodic disturbances using neural networks,'' \emph{IET Control Theory and
  Applications}, vol.~3, no.~10, pp. 1383--1394, 2009.

\bibitem{FSE4}
H.~Ma, H.~Ren, Q.~Zhou, R.~Lu, and H.~Li, ``Approximation-based nussbaum gain
  adaptive control of nonlinear systems with periodic disturbances,''
  \emph{IEEE Transactions on Systems, Man, and Cybernetics: Systems}, vol.~52,
  no.~4, pp. 2591--2600, 2022.

\bibitem{FSE5}
Y.~Zhao, F.~Tang, G.~Zong, X.~Zhao, and N.~Xu, ``Event-based adaptive
  containment control for nonlinear multiagent systems with periodic
  disturbances,'' \emph{IEEE Transactions on Circuits and Systems II: Express
  Briefs}, vol.~69, no.~12, pp. 5049--5053, 2022.

\bibitem{FSE7}
F.~Cheng, B.~Niu, L.~Zhang, and Z.~Chen, ``Prescribed performance-based
  low-computation adaptive tracking control for uncertain nonlinear systems
  with periodic disturbances,'' \emph{IEEE Transactions on Circuits and Systems
  II: Express Briefs}, vol.~69, no.~11, pp. 4414--4418, 2022.

\bibitem{FSE6}
H.~Ma, H.~Li, R.~Lu, and T.~Huang, ``Adaptive event-triggered control for a
  class of nonlinear systems with periodic disturbances,'' \emph{Science China
  Information Sciences}, vol.~63, no.~5, pp. 161--175, 2020.

\bibitem{linear1}
N.~Yang and J.~Li, ``New distributed adaptive protocols for uncertain nonlinear
  leader-follower multi-agent systems via a repetitive learning control
  approach,'' \emph{Journal of the Franklin Institute}, vol. 356, no.~12, pp.
  6571--6590, 2019.

\bibitem{lemma2}
M.~M. Polycarpou and P.~A. Ioannou, ``A robust adaptive nonlinear control
  design,'' \emph{Automatica}, vol.~32, no.~3, pp. 423--427, 1996.

\bibitem{fix_own}
X.~Wang, Z.~Li, Z.~He, and M.~V. Basin, ``Prescribed performance adaptive
  tracking control with small overshoot for a class of uncertain second-order
  nonlinear systems,'' \emph{IEEE Transactions on Circuits and Systems II:
  Express Briefs}, vol.~69, no.~9, pp. 3834--3838, 2022.

\bibitem{lemma3}
X.~Huang, W.~Lin, and B.~Yang, ``Global finite-time stabilization of a class of
  uncertain nonlinear systems,'' \emph{Automatica}, vol.~41, no.~5, pp.
  881--888, 2005.

\bibitem{lemma4}
B.~Cui, Y.~Xia, K.~Liu, and G.~Shen, ``Finite-time tracking control for a class
  of uncertain strict-feedback nonlinear systems with state constraints: A
  smooth control approach,'' \emph{IEEE Transactions on Neural Networks and
  Learning Systems}, vol.~31, no.~11, pp. 4920--4932, 2020.

\bibitem{FSEB}
S.~Liuzzo, R.~Marino, and P.~Tomei, ``Adaptive learning control of nonlinear
  systems by output error feedback,'' \emph{IEEE Transactions on Automatic
  Control}, vol.~52, no.~7, pp. 1232--1248, 2007.

\bibitem{differentiator1}
X.~Wang and B.~Shirinzadeh, ``Rapid-convergent nonlinear differentiator,''
  \emph{Mechanical Systems and Signal Processing}, vol.~28, pp. 414--431, 2012.

\end{thebibliography}

\end{document}